\documentclass[11pt]{article}
\usepackage[left=1in, top=1in, right=1in, bottom=1in]{geometry}

\usepackage{microtype}
\usepackage{graphicx}
\usepackage{subcaption}
\usepackage{amsmath, amssymb}
\usepackage{amsthm}
\usepackage{bbm}
\usepackage{mathtools}
\usepackage[numbers]{natbib}
\usepackage{color}
\usepackage[dvipsnames]{xcolor}
\usepackage{float}
\usepackage{tikz,pgfmath}
\usetikzlibrary{matrix,positioning,quotes}
\usetikzlibrary{shapes.multipart}
\usetikzlibrary{calc}
\usetikzlibrary{matrix}

\usepackage[backref,linktoc=all]{hyperref}
\hypersetup{
    colorlinks=false,
    linkcolor=blue,
    filecolor=magenta,      
    urlcolor=cyan,
}
\newcommand{\AutoAdjust}[3]{\mathchoice{ \left #1 #2  \right #3}{#1 #2 #3}{#1 #2 #3}{#1 #2 #3} }
\newcommand{\Xcomment}[1]{{}}

\newcommand{\InBrackets}[1]{\AutoAdjust{[}{#1}{]}}% {\left[{#1}\right]}
\newcommand{\Ex}[2][]{\operatorname{\mathbf E}_{#1}\InBrackets{#2}}

\newcommand{\Prx}[2][]{\operatorname{\mathbf{Pr}}_{#1}\InBrackets{#2}}

\newcommand{\eqdef}{\stackrel{\textnormal{def}}{=}}

\newcommand{\noaccents}[1]{#1}

\newcommand{\newagentvar}[3][\noaccents]{%
\expandafter\newcommand\expandafter{\csname #2\endcsname}{#1{#3}}%
\expandafter\newcommand\expandafter{\csname #2s\endcsname}{#1{\boldsymbol{#3}}}%
\expandafter\newcommand\expandafter{\csname #2smi\endcsname}[1][i]{#1{\boldsymbol{#3}}_{\text{-}##1}}%
\expandafter\newcommand\expandafter{\csname #2i\endcsname}[1][i]{#1{#3}_{##1}}%
\expandafter\newcommand\expandafter{\csname #2ith\endcsname}[1][i]{#1{#3}_{(##1)}}%
}

\DeclareMathOperator{\Conv}{Conv}
\DeclareMathOperator{\rank}{rank}

\newcommand{\matroid}{\mathcal M}

\newagentvar{val}{v}
\newagentvar[\tilde]{mulval}{v}
\newagentvar[\tilde]{multval}{v}
\newagentvar[\hat]{divval}{v}
\newagentvar{dist}{F}
\newagentvar{disttau}{G}

\newagentvar{thresh}{\theta}

\newcommand{\feasible}{\mathcal{F}}

\newcommand{\polytope}{\mathcal{P}}
\newagentvar{alloc}{x}
\newagentvar{avgalloc}{y}
\newagentvar{optalloc}{y}

\newcommand{\indicator}{\mathbbm{1}}

\newcommand\Hongxun[1]{}
\newcommand\hufu[1]{}
\newcommand\Abner[1]{}
\newcommand\del[1]{}

\theoremstyle{plain}
\newtheorem{theorem}{Theorem}[section]
\newtheorem{lemma}[theorem]{Lemma}

\newtheorem{definition}[theorem]{Definition}

\newcommand*\samethanks[1][\value{footnote}]{\footnotemark[#1]}

\title{Oblivious Online Contention Resolution Schemes}
\author{Hu Fu \thanks{ITCS, Shanghai University of Finance and Economics. {\texttt{\{fuhu, lu.pinyan, tang.zhihao\}@mail.shufe.edu.cn}}}
\and Pinyan Lu \samethanks[1]
\and Zhihao Gavin Tang\samethanks[1]
\and Abner Turkieltaub\thanks{University of British Columbia. {\texttt{abner7@cs.ubc.ca}}}
\and Hongxun Wu\thanks{IIIS, Tsinghua University. {\texttt{\{wuhx18, zqf18\}@mails.tsinghua.edu.cn}}}
\and Jinzhao Wu\thanks{CFCS, Computer Science Dept., Peking University. {\texttt{jinzhao.wu@pku.edu.cn}}}
\and Qianfan Zhang\samethanks[3]
}

\begin{document}
\maketitle

%\fancyfoot[R]{\scriptsize{Copyright \textcopyright\ 2022 by SIAM\\
%Unauthorized reproduction of this article is prohibited}}

    \begin{abstract}
  Contention resolution schemes (CRSs) are powerful tools for obtaining ``ex post feasible'' solutions from candidates that are drawn from ``ex ante feasible'' distributions.
  Online contention resolution schemes (OCRSs), the online version, have found myriad applications in Bayesian and stochastic problems, such as prophet inequalities and stochastic probing.
  
%  When the distribution is unknown, whether good CRSs/OCRSs exist with no sample (in which case the scheme is \emph{oblivious}) or few samples from the distribution is interesting, not only as a standard sample complexity question but also because of its potential connection to the matroid secretary problem, as recently proposed by Dughmi (ICALP 2020).
 
When the ex ante distribution is unknown, it was unknown whether good CRSs/OCRSs exist with no sample (in which case the scheme is \emph{oblivious}) or few samples from the distribution. 
In this work, we give a simple $\frac{1}{e}$-selectable oblivious single item OCRS by mixing two simple schemes evenly, and show, via a Ramsey theory argument, that it is optimal.
% We also give an $\Omega(1)$-selectable oblivious OCRS for laminar matroids.
On the negative side, we show that no CRS or OCRS with $O(1)$ samples can be $\Omega(1)$-balanced/selectable (i.e., preserve every active candidate with a constant probability) for graphic or transversal matroids.

% On the positive side, we give a $\frac 1 e$-selectable oblivious single item OCRS and show it is optimal via a Ramsey theory argument. We also give an $\Omega(1)$-selectable oblivious OCRS for laminar matroids.

  % On the positive sides, we give a $(\frac 1 4 - \eps)$-competitive prophet ineƒquality with $O(\log n)$ samples for general matroids, by way of a sample-based OCRS.
  % \Hongxun{Should we add "which relies on matroid OCRS with $O(\log n)$ samples" here?}
  
  %En route, we also give a $\frac 1 e$-selectable oblivious single item OCRS, and show it is optimal via a Ramsey theory argument.
  %For laminar matroids we give an $\Omega(1)$-selectable oblivious OCRS.

\end{abstract}

\section{Introduction}
\label{sec:intro}
% LIPics format requires bibliographystyle{plain}. Namely numerical references only. Thus I change all citet/citep to cite. 
Contention Resolution Schemes (CRSs) were introduced by Chekuri et al.~\cite{CVZ14} as a tool for rounding fractional solutions in submodular function maximization.
These schemes allow one to first optimize, under feasibility constraints, a continuous extension of a discrete function and then round the fractional solution to an integral feasible solution.
Feldman et al.~\cite{FSZ16} extended this framework to online settings, and the resulting Online Contention Resolution Schemes (OCRSs) turn out to be powerful tools for a wide range of applications in Bayesian and stochastic online optimization, such as prophet inequalities \cite{KW19, EFGT20}, stochastic probing \cite{GN13, GNS16, FTWWZ21}, and posted pricing mechanisms \cite{CHMS10}.

%Roughly speaking, a CRS is given an \emph{ex ante feasible} solution to a constrained optimization problem, and must produce \emph{ex post feasible} solutions.  
%The ex ante solution is in the form of Bernoulli distributions, one for each element, that describe the probabilities with which the elements appear in the ex ante solution.
\hufu{Should we be more specific here?  I spent a long time here and decided to give up for now..}
%\Hongxun{I found the original paragraph difficult to understand.. So I had a try below... }

% In a constrained optimization problem, we are sometimes given an \emph{ex ante feasible} solution, in the form of a distribution over feasible sets of candidates.
More concretely, given a feasibility system, an ex ante feasible solution is a distribution $\allocs$ over feasible sets.  
According to such a distribution, the elements are included in the solution in a correlated manner.  
In many problems, however, the elements are selectable, or \emph{active}, only \emph{independently} according to the marginal distributions given by~$\allocs$.  
A contention resolution scheme is a procedure indexed by the distribution~$\allocs$ that, given the set of active elements, 
% \emph{with the knowledge of these marginal distributions}, 
must select an (ex post) feasible subset of active elements, and aims to guarantee that each element, when active, is selected with a constant probability~$c$.
For many applications, this guarantees to retain at least a $c$~fraction of the objective, compared with that of the (unattainable) ex ante solution.
An online contention resolution scheme sees each element's status (of being active or not) in an online fashion, and must decide whether to select an element upon its arrival.

% In a discrete constrained optimization problem, we need to pick a feasible set of elements. 
% An \emph{ex ante feasible} solution is a vector that specifies the probability each element is picked. To round an ex ante solution, we let each element be \emph{active} (or selectable) with its probability. The set of active elements may not be feasible. Knowing such probabilities (i.e., the ex ante feasible solution), a CRS must pick a feasible subset of active elements to be the \emph{ex post feasible} solution. The goal of CRS is to guarantee that each active element is preserved with constant probability, which implies a constant approximation to the original constraint optimization problem.

% ~$\soldist$ over sets that are feasible solutions to a constrained optimization problem.
% Ideally, one would like the expected performance given by solutions sampled from~$\soldist$.
% Sampling sets, however, implies correlation among the elements. 
% Such correlation is often not allowed;
% instead, the elements are active (or selectable) \emph{independently} according to the marginal distributions given by~$\soldist$.
% A CRS observes the set of active elements, and must produce from them a feasible solution, such that each element, when active, is selected into the solution with a good probability in expectation.
% A CRS's performance is measured by this probability it guarantees for each element.
% The contention resolution scheme is online if it sees the elements' selectability one by one and must decide whether to select them in an online fashion.

%% prior independent 

In many Bayesian or stochastic problems, it is interesting to study whether good algorithms exist in the absence of the distributional knowledge, and, when sample access is allowed to make up for this lack of knowledge, how the performance of best algorithms scales with the number of samples.  
% the problem's sensitivity to the input distribution, as measured by the number of samples from the distribution required for the design of a good algorithm.
For example, in revenue optimal mechanism design, mechanisms with no knowledge of or only sample access to the type distributions are known as prior-free or prior-independent mechanisms, and has seen a large literature devoted to them \cite[e.g.][]{GHKSW06, HR08, DRY15, CR14, GW21}.  
As another example, in the single-item prophet inequality problem, a \emph{single} sample suffices for an algorithm to match the performance of an optimal one equipped with full knowledge of the distribution \cite{RWW20}. 
% in contrast, for the design of a revenue optimal auction, the number of required samples grows linearly with the number of bidders if one would like to closely approximate the optimal revenue \cite{CR14}. 
We investigate in this work similar problems for CRS and OCRS when the feasibility system is given by a matroid.
In particular, we give a simple, provably optimal, algorithm for single-item OCRS with no distributional knowledge ($\allocs$), and we show that no good CRS or OCRS with $O(1)$ samples exist for graphic or transversal matroids. 
% It was unknown before our work whether \emph{oblivious} contention resolution is possible in this setting, i.e., whether good contention resolution schemes exist without any sample from the distribution.

We now introduce more formalism in order to discuss our discoveries.

%% Dughmi

% Adding to the significance of this investigation, Dughmi~\cite{Dughmi20} recently raised the hope of making progress in, if not resolving, the long open matroid secretary problem \cite{CDFS19}, by proposing a blueprint reduction from the matroid secretary problem to a certain powerful OCRS.
% A key gap left open in this blueprint is that the reduction cannot provide the OCRS with a prior distribution.
% While Dughmi showed the impossibility of certain OCRS's that would fit in the reduction, he pointed out that the blueprint ``suggests that resolving contention with limited knowledge of the prior is closely related to the matroid secretary conjecture''; he explicitly raised the question whether a good CRS with no or few samples exists in the setting we study in this work.
% One of our main results answers the question in the negative.

%Section~\ref{sec:intro-CRS} discusses our results on contention resolution, and Section~\ref{sec:intro-prophet} discusses those on prophet inequalities.

\subsection{Oblivious CRS and OCRS}
\label{sec:intro-CRS}

% \paragraph{Contention Resolution Schemes (CRS)}

A CRS/OCRS is defined with respect to of a universe~$U$ and a downward-closed feasibility system $\feasible \subseteq 2^U$.  
An ex ante feasible solution is a distribution over sets in~$\feasible$, represented by a vector $\allocs$ in the polytope associated with~$\feasible$: $\allocs \in \polytope_{\feasible } \coloneqq \Conv(\{1_S\}_{S \in \feasible})$, where $\Conv(\cdot)$ is the operation of taking the convex hull, and $1_S$ is the indicator vector for a set~$S$. 
A convex decomposition of~$\allocs$ using the vertices of~$\polytope_{\feasible}$ gives rise to such a distribution,  and $\alloci$ is the marginal probability with which element~$i$ is included when one draws from this distribution.  
In many applications, either one is prevented from sampling from this distribution, or such a sample does not serve as a good solution; rather, one observes that each element~$i$ is selectable, or \emph{active}, \emph{independently}, with probability~$\alloci$.
A good CRS maps the set of active elements to a feasible subset, guaranteeing that each element, when active, is kept with a constant probability.

% Consider a matroid $\matroid = (U, \indset)$ with the universe $U = [n]$ and the set of independent sets $\indset \subseteq 2^U$. 
% Its polytope, the convex hull of the indicator vectors of its independent sets, is given by $\polytope_{\matroid} \coloneqq \{\allocs \in \mathbb [0, 1]^n \given \sum_{i \in S} \alloci \leq \rank(S), \forall S \subseteq U\}$.

\begin{definition}[\cite{CVZ14}]
  \label{def:CRS}
  A \emph{Contention Resolution Scheme} (CRS) $\pi$ for~$\feasible$ is a procedure indexed by ex ante feasible solutions $\allocs \in \polytope_\feasible$. 
  For every $\allocs \in \polytope_{\feasible}$, for each $S \subseteq U$, $\pi_{\allocs}(\cdot)$ returns a random set $\pi_{\allocs}(S)$ such that, with probability~1, $\pi_{\allocs}(S) \in \feasible$ and $\pi_{\allocs}(S) \subseteq S$.  
  A CRS $\pi$ is \emph{$c$-balanced}, for $c \in [0, 1]$, if for each $\allocs \in \polytope_{\feasible}$, for each $i \in U$, $i \in \pi_{\allocs}(R(\allocs))$ with probability at least $c \alloci$, where $R(\allocs)$ is a random subset of~$U$ such that each element $j$ is in~$R(\allocs)$ independently with probability~$\alloci[j]$.
  
  A CRS $\pi$ is said to be \emph{deterministic} if $\pi_{\allocs}(\cdot)$ is deterministic for each $\allocs$. 
  It is \emph{oblivious} if, for every $S \subseteq U$, 
  the distribution of $\pi_{\allocs}(S)$ and the distribution of $\pi_{\mathbf y}(S)$ are identical for any two $\allocs, \mathbf y \in \polytope_{\feasible}$.
  \footnote{In~\cite{CVZ14}, Chekuri et al.\@ additionally require an oblivious CRS to be deterministic.  
  We remove this requirement. 
  This distinction is important for all the results in this work.}
\end{definition}

The elements in the random set~$R(\allocs)$ are said to be \emph{active}.

  % A \emph{Contention Resolution Scheme} (CRS) takes as input a vector $\allocs \in \polytope_\feasible$ and a random set $R(\allocs)$ where each element $i \in U$ is in $R(\allocs)$ independently with probability $\alloci$; 
% an element~$i$ is \emph{active} if it appears in $R(\allocs)$.
% The CRS must then output $S \subseteq R(\allocs)$ such that $S$ is independent in the matroid.
% A CRS is \emph{$c$-balanced} (for $c \in [0, 1]$) if, for any $\allocs \in \polytope_{\matroid}$, each $i \in U$ is in the output~$S$ with probability at least $c\alloci$.

\begin{definition}[\cite{FSZ16}]
\label{def:OCRS}
An \emph{Online Contention Resolution Scheme} (OCRS) with respect to~$\feasible$ is a procedure indexed by ex ante feasible solutions $\allocs \in \polytope_{\feasible}$ and an arrival order of elements in~$U$.  
The elements arrive according to the order, and is revealed whether it is \emph{active} and, if so, the OCRS must decide irrevocably whether to accept it in the output~$T$. 
An OCRS is \emph{$c$-selectable}, for $c \in [0, 1]$, if for any $\allocs \in \polytope_{\matroid}$ and any arrival order, if each element~$i$ is active independently with probability $\alloci$, it is kept by the algorithm in the output with probability at least $c\alloci$.  

An OCRS is \emph{oblivious} if, for any arrival order, at any point in the procedure, given the set of elements that have arrived and those that have been accepted, and a newly arrived active element~$i$, the probability with which the algorithm accepts~$i$ is the same regardless of~$\allocs$.
\end{definition}

% A CRS or OCRS is \emph{oblivious} if the knowledge of $\allocs$ is irrelevant to its output.

In Definition~\ref{def:OCRS}, the arrival order may depend on~$\allocs$ but is fixed before the procedure without knowledge of the set of active elements.  
It is said to be from an \emph{offline adversary}.  
Both more powerful adversaries, such as online and almighty ones, and weaker adversaries, such as random order ones, have been studied in the literature;
we focus on the offline adversary model throughout this work.
Without loss of generality, we assume the elements arrive in the order $1, 2, \ldots, n$.

In the literature, especially that on OCRSs, the case where $\feasible$ is a matroid on~$U$ has been studied the most.  
We focus on this case in this paper, and denote the matroid by~$\matroid$.

As an illustration, consider the simplest OCRS for the rank-1 uniform matroid.
Here $\allocs$ satisfies $\sum_i \alloci \leq 1$ and each element $i$ is active independently with probability $x_i$. An OCRS may select at most one element.
It is easy to see that no deterministic oblivious CRS can be $\Omega(1)$-balanced even in this simple setting, as pointed out by Chekuri et al.~\cite{CVZ14}.   Consider the randomized OCRS that accepts an active element with probability~$\tfrac 1 2$ whenever nothing has been selected.
By the end, the algorithm selects nothing with probability at least~$\tfrac 1 2$, and hence at the arrival of each element, with probability at least $\tfrac 1 2$ nothing has been selected; the element, if active, therefore has a chance of at least~$\tfrac 1 4$ for being selected.
This OCRS is hence $\tfrac 1 4$-selectable.  
Note that it is oblivious, since nowhere does it make use of the distribution~$\allocs$.
An OCRS with knowledge of~$\allocs$ can fine-tune its probability of accepting an active element~$i$: if nothing has been accepted when $i$ arrives, it may accept~$i$ with probability $\frac 1 2 / {(1 - \frac 1 2 \sum_{j < i} \alloci[j])}$.   
By a simple induction, it can been that each element~$i$ is accepted with probability precisely $\frac 1 2$ conditioning on its being active, and with probability precisely $1 - \frac 1 2 \sum_{j < i} \alloci[j]$, the algorithm has accepted nothing when $i$ arrives.
This OCRS that first appeared in \cite{Alaei14} is $\tfrac 1 2$-selectable and is optimal for this setting.

% By default in this paper, we assume the arrival order for an OCRS is fixed by an adversary who does not know the realization of~$R(\allocs)$.

% The arrival order is called \emph{random} if it is uniformly at random, and the OCRS is called a \emph{random order OCRS} (RCRS);
% the order is \emph{adversarial} if it is chosen by an adversarial who does not know $R(\allocs)$ (but knows $\allocs$).
% the adversarial is \emph{almighty} if he knows $R(\allocs)$.
% By default, in this paper, an OCRS works with an adversarial order. 

%\zhihao{todo}

\subsection{Main Results}
\label{sec:results}

Our first main result is a simple $\frac 1 e$-selectable oblivious OCRS for the uniform rank~1 matroid, beating the $\frac 1 4$-selectable oblivious OCRS best known so far.
We also show that our scheme is the optimal oblivious OCRS in this setting.

% For oblivious OCRSs, can one improve upon the above $\frac 1 4$-selectable scheme?  
To motivate the design of our OCRS, consider the following simple, almost wild thought experiment: when processing element~$i$, the non-oblivious, $\frac 1 2$-selectable OCRS above uses information of~$\allocs$ only for the elements that come before~$i$; when we do not know~$\allocs$, what if we see the arrived elements as a (partial) sample from the underlying distribution, and emulate the optimal non-oblivious OCRS, as if the distribution is just the empirical distribution given by that sample?
To substantiate this thought, when the first active element~$i$ arrives, since no preceding element is active, we are led to crudely estimate $\sum_{j < i} \alloci[j]$ to be~$0$, and we accept~$i$ with probability $\frac 1 2$; 
when the second active element~$i'$ arrives, if we did not accept the first active element, then we would estimate $\sum_{j < i'} \alloci[j]$ to be~$1$, and we should accept~$i'$ with probability $\frac 1 2 / (1 - \frac 1 2) = 1$.  
To summarize, the resulting algorithm accepts the first active element with probability~$\frac 1 2$; otherwise, it accepts the second active element.
The scheme is obviously oblivious.

In Section~\ref{sec:rank1}, we show that precisely this OCRS is $\tfrac 1 e$-selectable.  
It may be a surprise that this simple algorithm is in fact optimal.  
Our proof of optimality uses Ramsey theory, following a strategy introduced by Correa et al.~\cite{CDFS19}.

% present a $\tfrac 1 e$-selectable oblivious OCRS for the single item setting, and show that it is optimal; that is, no oblivious OCRS can guarantee a selectability better than $\tfrac 1 e$. Our algorithm is an even mixture between the na\"ive greedy algorithm which accepts the first active element it sees, and another simple algorithm which rejects the first active element it sees and selects the second one (if there is one).   It may be a surprise that such an algorithm turns out optimal. We show its optimality by following a proof strategy via Ramsey theory, introduced by Correa et al.~\cite{CDFS19}.  
% Finally, we give an $\Omega(1)$-selectable OCRS for laminar matroids in Section~\ref{sec:laminar}.

Going beyond the uniform rank~1 matroids,  Chekuri et al.~\cite{CVZ14}, among other constructions, gave an oblivious $\Omega(1)$-balanced CRS for the unsplittable flow  problem in trees, which implies an oblivious $\Omega(1)$-balanced CRS for laminar matroids.  
For the general matroid setting,  it was not known prior to this work whether $\Omega(1)$-balanced/selectable oblivious CRS/OCRS exists. 
 We show in Section~\ref{sec:negative} that no oblivious CRS can be $\Omega(1)$-balanced even in graphic or transversal matroids.  
 This immediately also implies that no oblivious $\Omega(1)$-selectable OCRS exists in these settings.
Our proof shows that a good oblivious CRS must be able to distinguish, given the set of active elements, between a uniform distribution and a distribution with a hidden structure buried, which is statistically impossible. 
In fact, our construction shows that, even if a CRS has access to a constant number of samples of~$R(\allocs)$, it still cannot be $\Omega(1)$-selectable.
We note that these impossibility results are not computational, and are therefore not conditional on complexity assumptions.
% As we mention above, this answers in the negative a question recently posed by Dughmi~\cite{Dughmi20}, and suggests that it may be more challenging than it has been hoped to use contention resolution schemes to resolve the matroid secretary problem.

% \zhihao{On the positive side, we introduce a $\tfrac 1 e$-selectable oblivious OCRS for the single item setting in Section~\ref{sec:rank1},}

% \paragraph{Online Contention Resolution Schemes (OCRS)}
% An Online Contention Resolution Scheme (OCRS) is essentially a CRS that works in an online fashion.
% It takes as input $\allocs \in \polytope_\matroid$, and sees the membership in~$R(\allocs)$ of the elements in~$U$ arriving in order;
% when each element $i \in U$ arrives, if it is active (i.e., $i \in R(\allocs)$), the OCRS must decide irrevocably whether to include $i$ in its output~$S$.  
% Like a CRS, the output~$S$ must always be an independent set in~$\matroid$.
% The arrival order is called \emph{random} if it is uniformly at random, and the OCRS is called a \emph{random order OCRS} (RCRS);
% the order is \emph{adversarial} if it is chosen by an adversarial who does not know $R(\allocs)$ (but knows $\allocs$);
% the adversarial is \emph{almighty} if he knows $R(\allocs)$.
% By default, in this paper, an OCRS works with an adversarial order. 
% Without loss of generality, we assume the elements arrive in the order $1, 2, \ldots, n$.

\subsection{Further Related Works}
\label{sec:related}

CRSs were first developed in the submodular optimization literature \cite{CVZ14,CCPV07}, and are closely related to \emph{correlation gaps} \cite{ADSY12}.  
The application in submodular optimization requires an additional \emph{monotonicity} property on the CRS; both monotone and non-monotone CRSs are interesting objects of study in various settings (e.g. \cite{BZ19}).
As noted above, Chekuri et al.~\cite{CVZ14} defined oblivious CRSs with the additional requirement that they be deterministic, a constraint that we relax in our definition.  It is not difficult to see that no oblivious, deterministic $\Omega(1)$-balanced CRS exists even for choosing one element from a set of two.
Showing no good randomized oblivious CRS exists for matroids is more challenging.

% \Hongxun{Maybe here is incomplete.. "cite [e.g.] {BZ19}"}

OCRSs were formally defined by Feldman et al.~\cite{FSZ16}, although a problem that is equivalent (under disguise) had been studied by Alaei~\cite{Alaei14} for the uniform matroids.
Feldman et al.\@ gave a $\frac 1 4$-selectable OCRS for matroids.
Lee and Singla~\cite{LS18} obtained a $\frac 1 2$-selectable matroid OCRS by a reverse reduction to the matroid prophet inequalities, albeit requiring the latter to be competitive against the stronger ex ante optimal.
Ezra et al.~\cite{EFGT20} gave a $0.337$-selectable OCRS for bipartite matchings.
Adamczyk and Wlodarczyk~\cite{AW18} studied OCRSs when the elements arrive in a uniformly random order, and obtained, among other results, $\frac{1}{p+1}$-selectable schemes for $p$ matroid intersections in this setting.
Lee and Singla~\cite{LS18}'s reduction also gave a $(1 - \frac 1 e)$-selectable matroid OCRS with random arrival for the uniform rank~1 matroid.

Dughmi~\cite{Dughmi20, Dughmi21} in a recent series of two works showed that the celebrated matroid secretary problem~\cite{BIKK18} is equivalent to the existence of good \emph{universal} OCRSs for general matroids.  
Interestingly, before showing this equivalence, in the first paper~\cite{Dughmi20}, Dughmi conjectured oblivious OCRSs to be a stepping stone toward the matroid secretary problem. 
Even though this was bypassed in the final proof of the equivalence~\cite{Dughmi21}, it is suggestive that oblivious OCRSs may find other applications.

% Whereas it is relatively straightforward to obtain $c$-competitive prophet inequalities from $c$-selectable OCRSs, the reverse is less clear.
% \cite{LS18} showed that prophet inequalities that are $c$-competitive against the ex ante optimal do give rise to $c$-selectable OCRSs.
% As explained above, such reductions do not imply translations of sample complexity results between the two settings.

% \hufu{Say something about \cite{AKW19}, with another mention of \cite{RWW20}.}

% \section{Preliminaries}
% \label{sec:prelim}
% \input{prelim}

%\section{Oblivious OCRS}
%\label{sec:ocrs}

\section{An Optimal Oblivious Single Item OCRS}
\label{sec:rank1}
% \subsection{Algorithm}

In this section we give an optimal oblivious OCRS when at most one element can be accepted, i.e., for the uniform rank~1 matroid. 
This is known as the single item setting.
For a uniform rank~1 matroid $\matroid$ on the universe $U = [n]$, $\allocs \in \polytope_\matroid$ means $\alloci \geq 0$ for each~$i$, and $\sum_{i=1}^n x_i \le 1$. 
% We first define a family of greedy OCRS. We remark that the following OCRS applies to more general settings beyond the single item case. 
Feldman et al.~\cite{FSZ16} defined a general, simple family of OCRSs that they call \emph{greedy}.  
In the single-item case, only one such algorithm makes sense, and we term it \emph{the} greedy algorithm in this setting:

\begin{definition}[\cite{FSZ16}]
\label{def:greedy}
% For $b \in [0, 1]$, the \emph{$b$-Greedy OCRS} accepts with probability~$b$ any item that is active and still acceptable.   The $1$-Greedy OCRS is also called the Greedy OCRS.
The \emph{Greedy} algorithm accepts with probability~$1$ the first active element.
\end{definition}

% A folklore result says that the $\frac{1}{2}$-Greedy OCRS is oblivious and $\frac{1}{4}$-selectable in the single item case.
% It may be a surprise that, mixing the Greedy algorithm with the following, equally simple, Accept Second algorithm yields an OCRS with the optimal selectability for this setting.

Our next algorithm, to be mixed with the Greedy algorithm, is equally simple:

\begin{definition}
The \emph{Accept Second} algorithm passes the first active element and accepts the second active one.  
%It then rejects all the following items.
\end{definition}

% Accept Second OCRS is obviously oblivious.

Sections \ref{sec:rank1-opt} and \ref{sec:ramsey} prove the following two theorems, respectively.

\begin{theorem} \label{thm:OCRS:optimal:upper}
% In oblivious setting, there is a simple 
Running the Greedy algorithm and the Accept Second algorithm each with probability~$\frac 1 2$ is a
$\frac{1}{e}$-selectable oblivious single element OCRS. 
\end{theorem}

\begin{theorem}
\label{thm:opt-lb}
For any $\epsilon > 0$, no oblivious single-item OCRS is $(\frac{1}{e} + \epsilon)$-selectable. 
\end{theorem}
% The procedure is shown in Algorithm~\ref{alg:ObliviousRank1} which selects the first active item w.p. $1/2$ and the second one w.p. $1$. Note this algorithm is \textit{oblivious}. Despite its simplicity, it turns out to be surprisingly powerful, and achieves $\frac{1}{e}$-selectable. 

% \begin{algorithm}[H]
% \SetAlgoLined
% \label{alg:ObliviousRank1}
% \caption{Oblivious $\frac{1}{e}$-selectable single item OCRS}
 % \For{each arriving item $i$}{
 % \If {$i$ is the first active item}
  % {Select $i$ w.p. $1/2$}
 % \ElseIf {$i$ is the second active item}
    % {Select $i$ w.p. $1$}
% }
% \end{algorithm}

\subsection{Analysis of Selectability}
\label{sec:rank1-opt}

To prove Theorem \ref{thm:OCRS:optimal:upper}, we need to show that for every $i \in [n]$, $\Prx{\text{$i$ is accepted} \mid \text{$i$ is active}} \geq \frac{1}{e}$. Since our algorithm is an even mixture of Greedy and Accept Second, 
\[
\Prx{\text{$i$ is accepted} \mid \text{$i$ is active}} = \frac{1}{2} \cdot \Prx{\text{$i$ is the first or second active element} \mid \text{$i$ is active}}.
\]

Conditioning on $i$ being active, the probability that $i$ is the first active element equals $\prod_{j<i} (1-\alloci[j])$, and the probability that $i$ is the second active element equals $\sum_{j<i}\alloci[j] \prod_{k<i,k\neq j} (1-\alloc_k)$. Therefore,

\begin{align*}
\Prx{\text{$i$ is accepted} \mid \text{$i$ is active}}= \frac12 \left[\prod_{j<i} (1-\alloci[j]) + \sum_{j<i}\alloci[j] \prod_{k<i,k\neq j} (1-\alloc_k)\right].
\end{align*}

%\zhihao{shall we change notations from $x$ to $y$ in the following function and the lemma?}
For each $i \ge 0$, define $f_{i}(\allocs) \eqdef \prod_{j\le i} (1-\alloci[j]) + \sum_{j\le i}\alloci[j] \prod_{k\le i,k\neq j} (1-\alloc_k) $. We now lower bound the value of $f_i(\allocs)$ for all $i$ and all possible choices of $\allocs$.

\begin{lemma}
\label{lemma:rank1lower}
$f_n(\allocs) \ge (1-\frac{1}{n})^n + (1-\frac{1}{n})^{n-1}$ for all $\allocs \in \{\allocs \mid \forall i, x_i \ge 0; \sum_{i=1}^n \alloci \le 1 \}$, where the equality is achieved when $\allocs = (\frac{1}{n}, \frac{1}{n}, \cdots, \frac{1}{n})$.
% Define 
% \begin{align*}
% C_n & \coloneqq \min
% %_{\text{$\sum_{i=1}^n \alloci \le 1$,  $\alloc \geq 0$}}
% \prod_{i=1}^n (1-\alloc_i) + \sum_{i=1}^n \alloc_i \prod_{j\neq i} (1-\alloci[j]) \\
% \mathrm{s.t. } \quad & 
% \sum_{i = 1}^n \alloci \leq 1; \\
% & \alloci \geq 0, \quad \forall i.
% \end{align*}
% For any $n\ge 2$, we have $C_n = (1 - \frac{1}{n})^n + (1 - \frac{1}{n})^{n-1}$, 
% % In particular, $C_n$ is attained when 
% attained at $\allocs = (1/n, 1/n, \dots, 1/n)$. Therefore, $C_n \rightarrow \frac{2}{e}$ as $n \rightarrow \infty$. 
\end{lemma}

\begin{proof}
We prove the lemma by induction.
The base case when $n = 1$ is trivial.
For $n\ge 2$, since the domain space $\{\allocs \mid \forall i, x_i \ge 0; \sum_{i=1}^n \alloci \le 1 \}$ is compact and the function $f_n$ is continuous, we have that the minimum of $f_n$ is attained by  some $\allocs^* \in \{\allocs \mid \forall i, x_i \ge 0; \sum_{i=1}^n \alloci \le 1 \}$.

Let $\allocs^{\text{uniform}} = (\frac{1}{n}, \frac{1}{n}, \cdots, \frac{1}{n})$. We have $f_n(\allocs^{\text{uniform}}) = (1 - \frac{1}{n})^n + (1 - \frac{1}{n})^{n - 1}$.

Next, we prove that $f_n(\allocs^*) \ge f_n(\allocs^{\text{uniform}})$ by contradiction. 
Otherwise $\allocs^* \ne \allocs^{\text{uniform}}$ and there exists two indices $i, j$ such that $\alloci^{*} \neq \alloci[j]^{*}$. Fix $\allocs^{*}_{\text{-}\{i,j\}}$ and consider the following two quantities.

\[
p_0 = \prod_{k \in [n] \setminus \{i, j\}} (1-\alloci[k]); \qquad p_1 = \sum_{k \in [n] \setminus \{i, j\}} \alloci[k] \prod_{t \in [n] \setminus \{i, j, k\}} (1- \alloci[t]).
\]
%be the probability that none of the items besides $i, j$ is active while 
%be the probability that exactly one of these items besides $i,j$ is active. 

We rearrange the formula of $f_n(\allocs^*)$ as the following.
\begin{align}
f_n(\allocs^*) &= (1-\alloci^*) (1-\alloci[j]^*) \cdot p_0 + \alloci^*(1-\alloci[j]^*) \cdot p_0 + \alloci[j]^*(1-\alloci^*) \cdot p_0 + (1-\alloci^*)(1-\alloci[j]^*) \cdot p_1 \nonumber \\
& = (1 - \alloci^* \alloci[j]^*) \cdot p_0 + (1-\alloci^*)(1-\alloci[j]^*) \cdot p_1 \nonumber \\
            &= p_0 + p_1 - (\alloci^* + \alloci[j]^*) \cdot p_1 + \alloci^* \alloci[j]^* \cdot (p_1 - p_0). \label{eq:rearrange}
\end{align}
We now examine two cases based on the relation between $p_0$ and $p_1$. We construct $\allocs'$ by adjusting the two coordinates $\alloci^*$ and $\alloci[j]^*$ while keeping all other variables the same as $\allocs^{*}_{\text{-}\{i,j\}}$. 

\begin{itemize}
    \item If $p_0 > p_1$, let $\alloci' = \alloci[j]' = \frac{\alloci^* +\alloci[j]^*}{2}$. Since $\alloci^* \ne \alloci[j]^*$, we have $\alloci' \alloc_j' > \alloci^* \alloc_j^*$. By equation~\eqref{eq:rearrange}, we have $f_n(\allocs') < f_n(\allocs^*)$, that contradicts the optimality of $\allocs^*$.
    
    \item If $p_0 \leq p_1$, let $\alloci' = 0$ and $\alloci[j]' = \alloci^* + \alloci[j]^*$. Since $\alloci' \alloci[j]' = 0$ and by equation~\eqref{eq:rearrange}, $f_n(\allocs') \leq f_n(\allocs^*)$. Observe that when $\alloci' = 0$, $\allocs_{\text{-}i}'$ is an $(n-1)$-dimensional vector and $f_{n - 1}(\allocs'_{-i}) = f_n(\allocs')$. Applying the inductive hypothesis, we have
    \[
    \left(1-\frac{1}{n-1}\right)^{n-1} + \left(1-\frac{1}{n-1}\right)^{n-2} \le f_n(\allocs') \le f_n(\allocs^*) < f_n(\allocs^{\text{uniform}}) = \left(1-\frac{1}{n}\right)^n + \left(1-\frac{1}{n}\right)^{n-1},
    \]
    which contradicts the fact that $(1-\frac{1}{n})^n + (1-\frac{1}{n})^{n-1}$ is monotonically decreasing in $n$.
\qedhere
\end{itemize}
\end{proof}

Finally, since $(1-\frac{1}{n})^n + (1-\frac{1}{n})^{n-1}$ is monotonically decreasing in $n$, with the limit equal to $\frac{2}{e}$, we conclude that our algorithm is $\frac{1}{e}$-selectable.

\subsection{Optimality}
\label{sec:ramsey}
In this section, we prove Theorem~\ref{thm:opt-lb}.  We first provide a road map of our proof.
%\zhihao{ordinal strategy $\to$ 1) symmetric OCRS 2) cardinal OCRS 3) index-independent OCRS}
%\Hongxun{I vote for 2) cardinal OCRS}
\begin{itemize}
    \item We first define a restricted class of algorithms called \emph{counting-based} strategies and prove that no counting-based strategy is strictly better than $\frac{1}{e}$-selectable.
    \item Next, we prove that, for any oblivious OCRS, there must be a subset of elements on which it behaves like a counting-based strategy. 
    This is the most technical step and makes use of Ramsey theorem.
    \item Finally, we embed the hard instance into the subset and conclude that the hard instance for counting-based strategies applies to all oblivious algorithms.
    %\hufu{``Hardness'' instance suggests computational hardness to my ears.}
    %\Hongxun{Changed to "hard" instance}
\end{itemize}  

 We start with defining counting-based strategies. 

%\zhihao{When I see the definition, I interpret the probability $p_k$ as the following: if we have not selected any item yet before the $k$-th active item, we accept it with probability $p_k$. Hence, we do not necessarily have $\sum p_i \le 1$. I change the proof of Lemma~\ref{lemma:ordinal_optimality} accordingly.}
%\Hongxun{Thanks. That was a mistake. Sorry for that.}

\begin{definition}
An OCRS is a \emph{counting-based strategy} if it is fully characterized by an infinite sequence of probabilities $(p_1, p_2, \dots )$ as follows: when the algorithm sees the $k$-th active item, the algorithm accepts (and stops) with probability $p_k$.
\end{definition}

For instance, the algorithm in Theorem~\ref{thm:OCRS:optimal:upper} is a counting-based strategy characterized by the sequence $(\frac{1}{2}, 1, 0, \dots)$. We now show that it is optimal among all counting-based strategies.

\begin{lemma}\label{lemma:ordinal_optimality}
For every $\epsilon > 0$, there exists a sufficiently large $n$, such that for the uniform instance with $n$ items (i.e. $x_i= \frac{1}{n}$ for all $i\in[n]$), any counting-based strategy cannot achieve $\left(\frac{1}{e} + \epsilon\right)$-selectability. 
\end{lemma}
\begin{proof}
To prove this, it suffices to show that $\Prx{\text{$n$ is selected} \mid \text{$n$ is active}} \leq \frac{1}{e} + \epsilon$. 

Let $q_k$ be the probability of having exactly $k$ elements active before $n$. On the uniform instance, $$q_k = \binom{n - 1}{k} \frac{1}{n^k} \left(1 - \frac{1}{n}\right)^{n - 1 - k} \leq \frac{n^k}{k!} \cdot \frac{1}{n^k} \cdot \left(1 - \frac{1}{n}\right)^{n - 1 - k} = \frac{1}{k!} \left(1 - \frac{1}{n}\right)^{n - 1 - k}.$$

When $k \geq 3$, $q_k \leq \frac{1}{k!} < \frac{1}{e}$. When $k < 3$, $q_k \leq \left(1 - \frac{1}{n}\right)^{n - 1 - k} \rightarrow \frac{1}{e}$ when $n \rightarrow \infty$. Hence, there is a sufficiently large $n$ such that $q_k \leq \frac{1}{e} + \epsilon$ for all $k$. 
Then 
\begin{align*}
\Prx{\text{$n$ is selected} \mid \text{$n$ is active}} = & \sum_{k=1}^{n} q_{k-1} \prod_{i=1}^{k-1} (1-p_i) \cdot p_k \le \left(\frac{1}{e} + \epsilon\right) \sum_{k=1}^n \prod_{i=1}^{k-1} (1-p_i) \cdot p_k \\
= & \left( \frac{1}{e} + \epsilon \right) \cdot \left(1 - \prod_{i=1}^{n} (1-p_i)\right) \le \frac{1}{e} + \epsilon,
\end{align*}
which concludes the proof.
\end{proof}

Next, we prove that no oblivious OCRS can do better than counting-based strategies in the worst case even when the total number of items $N$ is known to our algorithm. 

Given $N$, any oblivious OCRS can be characterized by a function $f: 2^{[N]} \to [0,1]$ that specifies the behaviour of the OCRS. In particular, for each set $T \in [N]$, $f(T)$ represents the probability that we select the $(i= \max T)$-th item at step $i$ given that $T$ is the set of active items so far and we have not selected any item yet.

Intuitively, an oblivious OCRS has no information about the vector $\allocs$ and hence, its decision rule should not depend on the indices of the active items. Indeed, for a counting-based strategy, the corresponding function $f$ only depends on the size $|T|$ of the active set. That is, $f(T) = p_{|T|}$ for all $T$, where $\{p_i\}$ is the probability sequence of the counting-based strategy.

To formalize the intuition, we prove that for any oblivious OCRS, there exists a subset of items $S \subseteq [N]$ on which the algorithm performs like a counting-based strategy. We define $(\epsilon, n)$-approximate counting-based strategies on $S$.

\begin{definition} \label{def:eps-n-approx}
An oblivious OCRS $A$ is $(\epsilon,n)$-approximate to a counting-based strategy $O$ on $S$ if $f(T) \in [p_{|T|}, p_{|T|}+\epsilon]$ for all $T \subseteq S$ and $|T| \le n$, where $(p_1, p_2, \cdots)$ is the probability sequence of $O$.
\end{definition}

\begin{lemma}\label{lemma:approx_ordinal}
For any integer $n,m$ and $\epsilon > 0$, there exists a sufficiently large integer $N$ such that for any oblivious OCRS $A$, there exists a subset $S \subseteq [N]$ of size $m$ such that $A$ is $(\epsilon,n)$-approximate to a counting-based strategy on $S$. 
%For a general strategy $A$ and , there is an integer $N$ and a subset $S \subset [N]$ of size $|S| = n$. Suppose there is an instance $I$, . Then there exists an ordinal strategy $O$ whose selectability on $I$ differs at most $\epsilon$ with that of $A$. 
\end{lemma}

% We defer the proof of the above lemma to the end of this subsection and show how it implies the optimality of our OCRS.
We first use the lemma to give a proof of Theorem~\ref{thm:opt-lb}, before proving the lemma itself.

\paragraph{Proof of Theorem~\ref{thm:opt-lb}}
By Lemma~\ref{lemma:approx_ordinal}, for any integer $n$ and $\epsilon>0$, there exists an integer $N$ such that for any oblivious OCRS $A$, there exists $S \subseteq [N]$ of size $n$ such that $A$ is $(\epsilon,n)$-approximate to a counting-based strategy $O$ on $S$.

Let $(p_1,p_2, \cdots)$ be the corresponding probability sequence of $O$.
Consider any instance that is supported on $S$, i.e. $x_i = 0$ for all $i \not\in S$. We have that for each $i \in S$,
\begin{align*}
& \Prx{A \text{ selects } i \mid i \text{ is active}} \\
={} & \sum_{T \ni i; T \subseteq S} \Prx{T-i \text{ are the active items before } i \text{ and are not selected by } A} \cdot f(T) \\
\le{} & \sum_{T \ni i; T \subseteq S} \Prx{T-i \text{ are the active items before } i \text{ and are not selected by } O} \cdot (p_{|T|} + \epsilon) \\
\le{} & \Prx{O \text{ selects } i \mid i \text{ is active}} + \epsilon,
\end{align*}
where the first inequality follows from the fact that whenever we have a chance to select an item, $O$ selects it with smaller probability than $A$ (by Definition \ref{def:eps-n-approx}).

Therefore, if $A$ is $(\frac{1}{e}+2\epsilon)$-selectable, $O$ is $(\frac{1}{e}+\epsilon)$-selectable for all instances defined on $S$. However, this contradicts Lemma \ref{lemma:ordinal_optimality}, that states any counting-based strategy is no better than $(\frac{1}{e}+\epsilon)$-selectable for the uniform instance on $S$, i.e. $x_i = \frac{1}{n}$ for $i\in S$. This concludes the proof of the theorem.

\paragraph{Proof of Lemma~\ref{lemma:approx_ordinal}.}
We use the Hypergraph Ramsey theorem. For completeness, we include a few basic concepts of hypergraphs.
A hypergraph is called a $k$-uniform hypergraph if each of its edges contains $k$ vertices. A complete $k$-uniform hypergraph is a hypergraph $G = (V, E)$ where $E$ is the set of all size-$k$ subsets of $V$. A clique in a $k$-uniform hypergraph is a subset $S \subset V$ such that all size-$k$ subsets of $S$ are in~$E$. 

\begin{lemma}[Hypergraph Ramsey Theorem \cite{ramsey1930problem,erdos1952combinatorial}] \label{lemma:ramsey}
Given any positive integer $n_0, k$ and $c$, there is an integer $n_1$ which has the following property: for any complete $k$-uniform hypergraph with more than $n_1$ vertices, no matter how we color its edges with $c$ colors, there is always a monochromatic clique of size $n_0$. 
\end{lemma}

We now prove Lemma~\ref{lemma:approx_ordinal} by induction. For the base case when $n=1$, consider a graph $G$ with $N = \left(\lfloor \frac{1}{\epsilon} \rfloor + 1\right) m$ vertices. Given an oblivious OCRS, we color each vertex $i$ with color $\lfloor \frac{1}{\epsilon} f(\{i\}) \rfloor$.
Since there are at most $\lfloor \frac{1}{\epsilon}\rfloor  + 1$ different colors, by the pigeonhole principle, there are at least $m$ vertices sharing the same color.
Let $S$ be the set of these $m$ vertices and $c$ be their color. We claim that $A$ is $(\epsilon,1)$-approximate to any counting-based strategy on $S$ with $p_1=c\epsilon$.

Suppose the lemma holds for $n-1$, we consider the case of $n$. Applying Lemma~\ref{lemma:ramsey} to $n_0=m, k=n$ and $c=\lfloor\frac{1}{\epsilon}\rfloor + 1$, let $n_1$ be the sufficiently large number such that for any complete $k$-uniform hypergraph with more than $n_1$ vertices, no matter how we color its edges with $c$ colors, there exists a monochromatic clique of size $n_0$.

By induction hypothesis, there exists $N$ such that for any oblivious OCRS $\pi$, there exists a subset $S_0 \subseteq [N]$ of size $n_1$ such that $\pi$ is $(\epsilon, n-1)$-approximate to a counting-based strategy on $S_0$. Now we consider a complete $n$-uniform hypergraph $G$ whose vertices are $S_0$. For each hyperedge $T \subset S_0$ of size $n$, we color the edge with color $\lfloor \frac{1}{\epsilon} f(T)\rfloor$. Observe that there are at most $\lfloor \frac{1}{\epsilon} \rfloor + 1$ different colors. By Lemma~\ref{lemma:ramsey}, this hypergraph includes a monochromatic clique of size $m$. We denote the vertex set of this clique by $S$ and let $c$ be its corresponding color. By definition of the hypergraph, we have that $f(T) \in [c\epsilon, (c+1)\epsilon]$ for $T \subset S$ and $|T|=n$. Combining this with the inductive hypothesis that $\pi$ is $(\epsilon,n-1)$-approximate to a counting-based strategy on $S \subseteq S_0$, we conclude that $\pi$ is further $(\epsilon,n)$-approximate to a counting-based strategy on $S$ with $p_{|T|} = c\epsilon$.

\section{Non-existence of $\Omega(1)$-selectable Oblivious CRSs for Matroids}
\label{sec:negative}
In this section we show that, without knowledge of~$\allocs$, $\Omega(1)$-selectable OCRS is impossible for general matroids. 
In fact, even in the offline setting, $\Omega(1)$-balanced CRS does not exist for graphic or transversal matroids. Furthermore, no such CRS exists even with a constant number of samples.
We provide the proofs for graphic matroids and transversal matroids in Section~\ref{subsec:graphic} and Section~\ref{subsec:transversal}, respectively.

% \zhihao{Zhihao: Do we have the updated proof below? I have not read them.}

\begin{theorem} \label{thm:non_existence_graphical}
  For any $c \in (0, 1]$, there is no oblivious $c$-balanced CRS for graphic matroids or transversal matroids.
Moreover, the impossibility persists even if the CRS has access to $O(1)$ samples of $R(\allocs)$. 
\end{theorem}

% The intuition behind the proof is to exhibit a matroid where different points in its matroid polytope generates random sets that looks the same and are far from being independent.   

% \subsubsection{Graphic Matroids}
% Given a graphical matroid $(E,\mathcal{I})$ determined by a graph $G=(V,E)$, we can characterize its matroid polytope as follows:
% $$\mathcal{P}=\{x\in [0,1]^E \ : \ \sum_{e\in E(U)} x_e \leq \rank(E(U)), \ \forall U\subseteq V\}$$
% Where $E(U)$ are the edges with both endpoints in $U$.

\subsection{Graphic Matroids}
\label{subsec:graphic}

Recall that a graphic matroid $(E,\mathcal{I})$ is defined on an undirected graph $G$ with edge set~$E$ such that $I \subseteq E$ is in $\mathcal{I}$ if{f} $I$ is a forest in $G$.

Consider a complete bipartite graph $K_{N,M}$ with bipartition $U=\{u_1,\ldots,u_N\}$ and $V=\{v_1,\ldots,v_M\}$. 
Let $\delta(u)$ denote the set of edges incident to a vertex $u \in U \cup V$. Let $\polytope$ be the polytope of the graphic matroid on it.

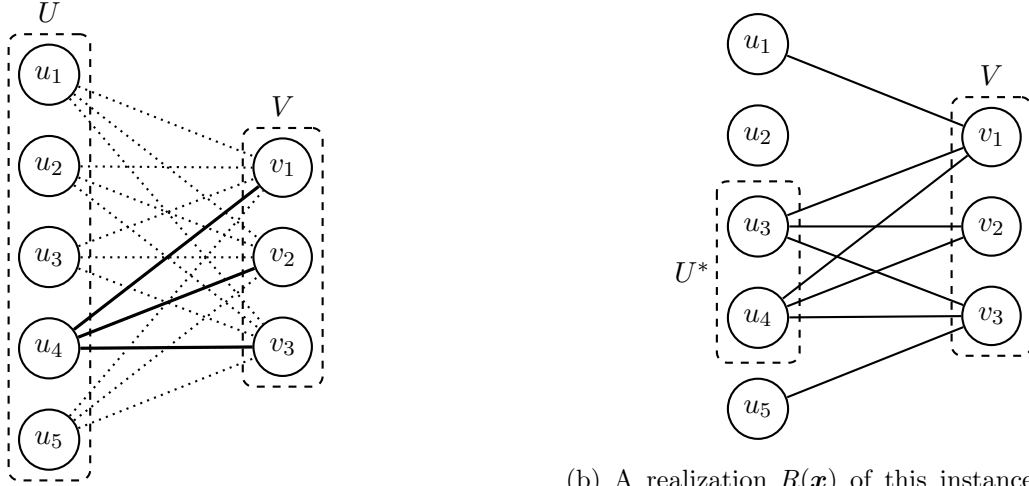
\begin{figure}[H]
\centering
\begin{subfigure}[b]{0.45\textwidth}
\centering
\begin{tikzpicture}[thick,amat/.style={matrix of nodes,nodes in empty cells,
  fsnode/.style={draw,solid,circle,execute at begin node={$u_{\the\pgfmatrixcurrentrow}$}},
  ssnode/.style={draw,solid,circle,execute at begin node={$v_{\the\pgfmatrixcurrentrow}$}}}]

 \matrix[amat,nodes=fsnode,label=above:$U$,row sep=1em,dashed,draw,rounded corners] (mat1) {
 \\
 \\ 
 \\
 \\
 \\};

 \matrix[amat,right=2cm of mat1,nodes=ssnode,label=above:$V$,row sep=1em,dashed,draw,rounded corners] (mat2) {\\
 \\ 
 \\};
\foreach \x in {1,...,5}
    \foreach \y in {1,...,3} 
    {\draw  (mat1-\x-1) edge[thick, dotted] (mat2-\y-1); } 

\foreach \y in {1,...,3} 
	{\draw  (mat1-4-1) edge[very thick] (mat2-\y-1); }
\end{tikzpicture}
\caption{The bipartite complete graph $K_{N,M}$. Here $i = 4$, edges adjacent to $u_i$ has probability $x_e^i = 1$ of being active, while other edges each only has probability $1/M$ of being active. Here $N$ should be a large enough number such that $N \gg M^M$.}
\end{subfigure}
\hfill
\begin{subfigure}[b]{0.45\textwidth}
\centering
\begin{tikzpicture}[thick,amat/.style={matrix of nodes,nodes in empty cells,
  fsnode/.style={draw,solid,circle,execute at begin node={$u_{\the\pgfmatrixcurrentrow}$}},
  ssnode/.style={draw,solid,circle,execute at begin node={$v_{\the\pgfmatrixcurrentrow}$}}}]

 \matrix[amat,nodes=fsnode,row sep=1em] (mat1) {
 \\
 \\ 
 \\
 \\
 \\};

 \matrix[amat,right=2cm of mat1,nodes=ssnode,label=above:$V$,row sep=1em,dashed,draw,rounded corners] (mat2) {\\
 \\ 
 \\};
\draw  (mat1-1-1) edge (mat2-1-1); 
%\draw  (mat1-3-1) edge (mat2-3-1); 
\draw  (mat1-5-1) edge (mat2-3-1); 

\draw[thick,dashed,rounded corners] ($(mat1-3-1)+(-0.55,0.6)$)  rectangle node[left, xshift = -0.5cm] {$U^*$} ($(mat1-4-1)+(0.55,-0.6)$);

\foreach \y in {1,...,3} 
	{\draw  (mat1-3-1) edge (mat2-\y-1); \draw  (mat1-4-1) edge (mat2-\y-1); }
\end{tikzpicture}
\caption{A realization $R(\allocs)$ of this instance. $U^*$ is the set of all vertices on left side of degree $M$. If $N$ is large enough there will be many vertices happen to be in $U^*$. These vertices in $U^*$ are indistinguishable to CRS, and $u_4$ ($i = 4$) is hidden between them.}
\end{subfigure}
\caption{The hard instance for graphic matroids}
\label{fig:Graphical}
\end{figure}

We give $N+1$ points $\allocs^1, \ldots, \allocs^N, \avgallocs$ in the matroid polytope $\polytope$, and show that, for any $c \in (0, 1]$, there are large enough $N$ and~$M$, such that no oblivious CRS can be $c$-balanced for the graphic matroid on the complete bipartite graph~$K_{N,M}$. The point $\avgallocs$ is auxiliary in the proof.

For any $i \in [N]$, let $\allocs^i$ be 
% we claim that the following point $\allocs^i$ is in $\mathcal{P}$:
$$
\alloc^i_e  \coloneqq \begin{cases}
1, & \text{if } e \in \delta(u_i);  \\ 
1/M, & \text{otherwise.} 
\end{cases}
$$

Let $\avgallocs$ be the vector with weight $\tfrac 1 M$ on all edges, i.e., $y_e = \frac 1 M$ for all $e$.
Note that $\avgallocs = \frac N {N+M-1} \frac 1 N \sum_{i = 1}^N \allocs^i$.

We first verify $\allocs^i \in \polytope$ for each~$i$, which immediately implies $\avgallocs \in \polytope$ as well.
Let $T_{i, j}$ be the tree whose edge set is $\delta(u_i) \cup \delta(v_j)$.
%consists of all the edges incident to~$u_i$ and all the edges incident to~$v_j$.  
It is easy to verify that $\allocs^i$ is the average of the indicator vectors of $T_{i, j}$ as $j$ ranges from 1 to~$M$; that is, $\allocs^i = \frac 1 M \sum_{j = 1}^M \indicator_{T_{i, j}}$, where $\indicator_{T_{i,j}}$ is the indicator vector of~$T_{i, j}$.
Finally consider the random set $U^*$ defined as follows: 
$$U^*=\{u_i\in U \mid \delta(u_i)\subseteq R(\avgallocs)\}.$$
For the sake of a contradiction, suppose there is a $c$-balanced CRS $\pi$ for $\polytope$. 
We analize the CRS on $R(\avgallocs)$ by considering the expected number of elements accepted from the set of edges incident to $U^*$ (i.e. $\delta(U^*)$). We will first give an upper bound using the feasibility constraints. Then we will provide a lower bound using the assumption that the CRS is $c$-balanced to get a contradiction. 
Let $\pi(R(\avgallocs))$ be the set of accepted elements, then since the output must be an independent set, we have:
$$\Ex {|\pi(R(\avgallocs))\cap \delta(U^*)|}\leq \Ex {\rank(\delta(U^*))}.$$
And the rank of $\delta(U^*)$ is just $|U^*|+|V|-1$. Observe that $|U^*|$ follows a binomial distribution with parameters $N, M^{-M}$ (each vertex $u_i$ belongs to $U^*$ with probability $M^{-M}$ independently). So we get the following upper bound: 
$$\Ex {|\pi(R(\avgallocs))\cap \delta(U^*)|}\leq \frac{N}{M^M}+M-1.$$
On the other hand, $\Ex {|\pi(R(\avgallocs))\cap \delta(U^*)|}$ can be rewritten as: $$\sum_{e\in E} \Prx{e\in \pi(R(\avgallocs)) \mid e\in \delta(U^*)} \Prx{e\in \delta(U^*)}.$$
Take an arbitrary edge $e=(u_i,v_j)$. First note that the probability of $e\in \delta(U^*)$ is just the probability of $u_i\in U^*$, so it is equal to $M^{-M}$. The crucial observation is that the distribution of $R(\avgallocs)$ conditioned on $(u_i,v_j)\in \delta(U^*)$ is identical to the distribution of $R(\alloc^i)$, so using the assumption that $\pi$ is $c$-balanced and that $(u_i,v_j)$ is always active for $x^i$, we get:
$$\Prx{(u_i,v_j)\in \pi(R(\avgallocs)) \mid (u_i,v_j)\in \delta(U^*)} = \Prx{(u_i,v_j)\in \pi(R(\alloc^i))} \geq c.$$ 
Since there are $NM$ edges in total, we obtain the following lower bound: 
$$\Ex {|\pi(R(\avgallocs))\cap \delta(U^*)|} \geq \frac{cN}{M^{M-1}}.$$
By putting together our upper and lower bounds, we get that:
$$c\leq \frac{1}{M}+\frac{M^{M-1}(M-1)}{N},$$
which can be arbitrarily small for $N\gg M^M$ and $M\gg 1$. This finishes the proof in the oblivious case. \\

For the case in which the CRS has access to a constant number of samples, the idea is almost the same. The only difference is that now we define $U^*$ as the set of vertices $u$ such that $\delta(u)$ is contained in every sample from $R(\avgallocs)$ and in the final realization of $R(\avgallocs)$. If we have $s$ samples, then each vertex has a probability $M^{-(s+1)M}$ of being in $U^*$. The same arguments hold, we just need $N\gg M^{(s+1)M}$ for the contradiction.

\subsection{Transversal Matroids}
\label{subsec:transversal}

Recall that a transversal matroid $\matroid=(L,\mathcal{I})$ can be defined by a bipartite graph $G=(L\sqcup R, E)$ such that $I\in \mathcal{I}$ iff there is a matching in $G$ covering $I$. The proof is very similar to the one provided for graphic matroids.
%\begin{theorem}
%  \label{thm:non_existence_transversal}
%There is no oblivious $c$-balanced CRS for arbitrary transversal matroids for any constant $c$. Moreover, the impossibility persists even if we have access to a constant number of samples from the distribution.
%\end{theorem}

Consider the transversal matroid defined by the bipartite graph $G=(L\sqcup R, E)$ (see Figure \ref{fig:Transversal_matroid}), where
\begin{itemize}
    \item $L=\{l_{i,j} \mid i\in [N], j\in [M]\}$,
    \item $R=\{r_i \mid i\in [N]\} \cup \{u_k \mid k\in [M-1]\}$,
    \item $E=\{(l_{i,j},r_i) \mid i\in [N], j\in [M]\} \cup \{(l_{i,j},u_k) \mid i\in [N], j\in [M], k\in [M-1]\}$.
\end{itemize}

In other words, $L$ contains $N$ sets of $M$ vertices, namely $L=\bigcup_{i\in [N]} L_i$ with $L_i=\{l_{i,j} \mid j\in [M]\}$; for each $L_i$ there is a vertex $r_i$ fully connected to it. Additionally, there is a set of $M-1$ vertices in $R$ fully connected to $L$. 
%Denote $L_i:=\{l_{i,j} \mid j\in [M]\}$ 
%and define for $S\subseteq L$ the set of indices $Z(S):=\{i\in [N] \mid S\cap L_i\neq \emptyset\}$. It is not hard to see that for any $S\subseteq L$, $\rank(S)=|Z(S)|+\min(M-1, |S|-|Z(S)|)$. 
Now, we claim that for any $i' \in [N]$, the following point $\allocs^{i'}$ belongs to the matroid polytope: 
$$\alloc^{i'}_{l_{i,j}}\coloneqq\begin{cases}
1, & \text{if } i=i';  \\ 
1/M, & \text{otherwise}. 
\end{cases}$$

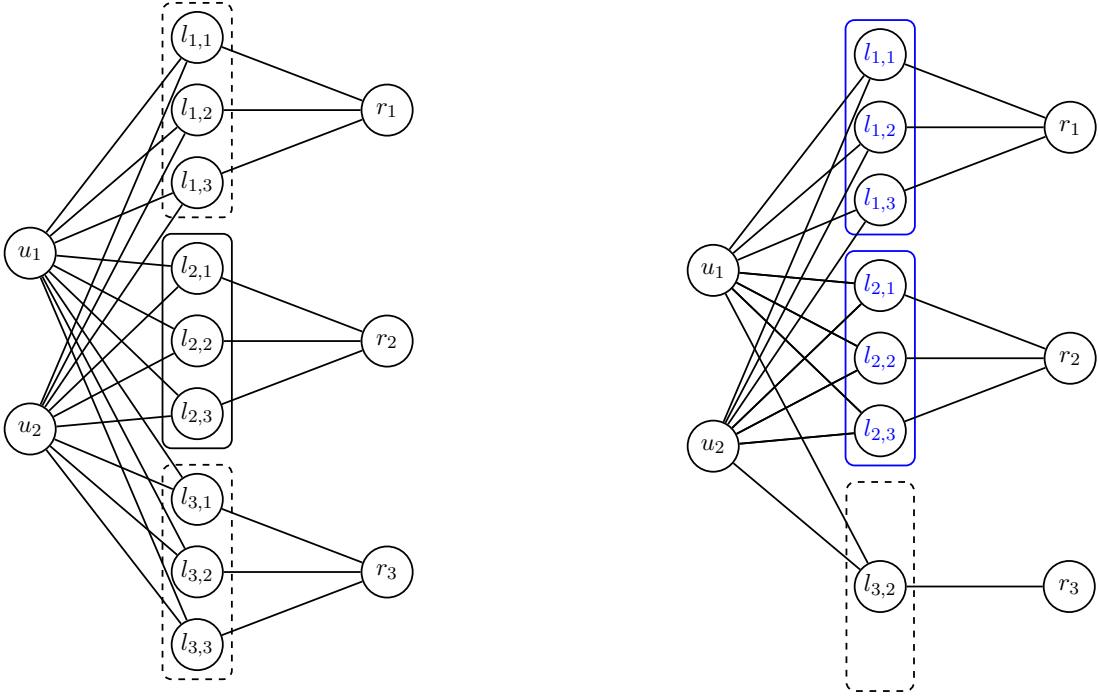
\begin{figure}[H]
\centering
\begin{subfigure}{0.45 \linewidth}
\centering
\scalebox{0.85}{
\begin{tikzpicture}[thick,amat/.style={matrix of nodes,nodes in empty cells}]

 \matrix[amat,nodes={inner sep=0pt, minimum size = 0.8cm, draw,solid,circle,execute at begin node={$l_{1, \the\pgfmatrixcurrentrow}$}},row sep=0.8em,dashed,draw,rounded corners] (mat1) {
 \\
 \\ 
 \\};

  \matrix[amat,nodes={inner sep=0pt, minimum size = 0.8cm, draw,solid,circle,execute at begin node={$l_{2, \the\pgfmatrixcurrentrow}$}},row sep=0.8em,thick,draw,rounded corners,below = of mat1, yshift = 2em] (mat2) {
 \\
 \\ 
 \\};

  \matrix[amat,nodes={inner sep=0pt, minimum size = 0.8cm, draw,solid,circle,execute at begin node={$l_{3, \the\pgfmatrixcurrentrow}$}},row sep=0.8em,dashed,draw,rounded corners, below = of mat2, yshift = 2em] (mat3) {
 \\
 \\ 
 \\};

 \node (r1) [draw,solid,circle,execute at begin node={$r_1$}, right = 2cm of mat1, minimum size = 0.8cm] {};

 \node (r2) [draw,solid,circle,execute at begin node={$r_2$}, right = 2cm of mat2, minimum size = 0.8cm] {};

 \node (r3) [draw,solid,circle,execute at begin node={$r_3$}, right = 2cm of mat3, minimum size = 0.8cm] {};

\foreach \x in {1,...,3} 
	{\draw  (mat1-\x-1) edge (r1); \draw  (mat2-\x-1) edge (r2); \draw  (mat3-\x-1) edge (r3); }

 \matrix[amat,nodes={inner sep=0pt, minimum size = 0.8cm, draw,solid,circle,execute at begin node={$u_{\the\pgfmatrixcurrentrow}$}},row sep=5em, left = 1.5cm of mat2] (mat4) {
 \\
 \\};

 \foreach \x in {1,2}
 	\foreach \k in {1,2,3}
 		\foreach \y in {1,2,3}
 			{\draw (mat4-\x-1) edge (mat\k-\y-1);}

\end{tikzpicture}
}
\caption{The bipartite graph $G=(L\sqcup R,E)$. There are $N$ sets in $L$ with $M$ vertices each. Here $i' = 2$. For each vertex $l_{2,j}$ in the second set, $x_{l_{2, j}} = 1$, while vertices in the other sets each only has $1/M$ probability of being active. The vertices in set $R = \{u_1, \dots, u_{M - 1}\} \cup \{r_1, \dots, r_N\}$ are drawn on two separated sides.}
\label{fig:Transversal_matroid}
\end{subfigure}
\hfill
\begin{subfigure}{0.45\linewidth}
\centering
\scalebox{0.85}{
\begin{tikzpicture}[thick,amat/.style={matrix of nodes,nodes in empty cells}]

 \matrix[amat,nodes={inner sep=0pt, minimum size = 0.8cm, draw,solid,circle,execute at begin node={$l_{1, \the\pgfmatrixcurrentrow}$}},row sep=0.8em,blue,draw,rounded corners] (mat1) {
 \\
 \\ 
 \\};

  \matrix[amat,nodes={inner sep=0pt, minimum size = 0.8cm, draw,solid,circle,execute at begin node={$l_{2, \the\pgfmatrixcurrentrow}$}},row sep=0.8em,blue,draw,rounded corners,below = of mat1, yshift = 2em] (mat2) {
 \\
 \\ 
 \\};

  \matrix[amat,nodes={inner sep=0pt, minimum size = 0.8cm, draw = none},row sep=0.8em,dashed,draw,rounded corners, below = of mat2, yshift = 2em] (mat3) {
 \\
 \\ 
 \\};

 \node (r1) [draw,solid,circle,execute at begin node={$r_1$}, right = 2cm of mat1, minimum size = 0.8cm] {};

 \node (r2) [draw,solid,circle,execute at begin node={$r_2$}, right = 2cm of mat2, minimum size = 0.8cm] {};

 \node (r3) [draw,solid,circle,execute at begin node={$r_3$}, right = 2cm of mat3, minimum size = 0.8cm] {};

 \node (r4) [inner sep=0pt, draw,solid,circle,execute at begin node={$l_{3,2}$}, at = (mat3-2-1), minimum size = 0.8cm] {};

\foreach \x in {1,...,3} 
	{\draw  (mat1-\x-1) edge (r1); \draw  (mat2-\x-1) edge (r2);}

 \matrix[amat,nodes={inner sep=0pt, minimum size = 0.8cm, draw,solid,circle,execute at begin node={$u_{\the\pgfmatrixcurrentrow}$}},row sep=5em, left = 1.5cm of mat2] (mat4) {
 \\
 \\};

 \foreach \x in {1,2}
 	\foreach \k in {1,2,2}
 		\foreach \y in {1,2,3}
 			{\draw (mat4-\x-1) edge (mat\k-\y-1);}

\draw (mat4-1-1) edge (r4);
\draw (mat4-2-1) edge (r4);
\draw (r4) edge (r3);
\end{tikzpicture}
}
\caption{A realization $R(\allocs^{i'})$ of this instance. A set belongs to $U^*$ if and only if all vertices in it are active. They are marked in blue. As long as $N$ is large enough, there will be many sets in $U^*$ with high probability. These sets are indistinguishable to our CRS. Hence the $i'$-th set is hidden inside. }
\label{fig:Transversal_polytope}
\end{subfigure}
\caption{The hard instance for transversal matroids}
\label{fig:Transversal}
\end{figure}

To see that $\allocs^{i'}\in \polytope_{\matroid}$, 
% we can take any $S\subseteq L$ and note that: 
%\begin{itemize}
%\item If $i^*\notin Z(S)$: $\ \displaystyle \sum_{l_{i,j}\in S} \allocs^{i^*}_{l_{i,j}} = \sum_{i\in Z(S)}\sum_{l_{i,j}\in S\cap L_i} M^{-1} \leq |Z(S)|$.
%\item If $i^*\in Z(S)$: $\ \displaystyle \sum_{l_{i,j}\in S} \allocs^{i^*}_{l_{i,j}} \leq (|Z(S)|-1) + |S\cap L_{i^*}|\leq (|Z(S)|-1)+\min(M, |S|-(|Z(S)|-1))$.
%\end{itemize}
%In any case, $\sum_{l_{i,j}\in S} \allocs^{i^*}_{l_{i,j}}\leq \rank(S)$. \\
% I like the above proof better, but the following one is more similar to what we have done previously. 
define for each $k\in [M]$ the set $T_{i',k}\coloneqq\{l_{i',j} \mid j\in [M]\} \cup \{l_{i,k} \mid i\in[N]\}$. It is easy to see that $T_{i',k}$ is independent (we can match each $l_{i,k}$ with $r_i$ and use the vertices $u_j$ to match the $M-1$ remaining vertices). Notice that $\allocs^{i'}$ is the average of the indicator vectors of $T_{i',k}$, that is $\allocs^{i'}=\frac{1}{M}\sum_{k\in [M]} \indicator_{T_{i',k}}$, so $\allocs\in \polytope_{\matroid}$.
Again, let $\avgalloc$ be the vector with weight $\frac{1}{M}$ on all vertices from $L$. We can see that $\avgalloc\in \polytope_{\matroid}$ since $\avgalloc\leq \allocs^{i'}$. The rest of proof is almost identical to the previous proof. Instead of $U^*$, we define now the random set of indices: 
$$I^*=\{i\in [N] \mid L_i \subseteq R(\avgalloc)\},$$
and then we compute bounds for the number of accepted elements from $L^*=\bigcup_{i\in I^*}L_i$.
The rank of $L^*$ is $|I^*|+M-1$: for each $i\in I^*$ we can match one vertex from $L_i$ to $r_i$, and then we can match $M-1$ more vertices from $L^*$ using the vertices $u_t$ (it is not possible to match more vertices since the rest of the vertices from $R$ are not connected with $L^*$). Note that the size of $I^*$ also follows a binomial distribution with parameters $N$ and $M^{-M}$, so we get the exact same upper bound we got before. We also get the same lower bound using equivalent arguments and that concludes the proof.

%\subsection{Oblivious OCRS for Laminar Matroids}
%\label{sec:laminar}
%\input{laminar}

\bibliographystyle{plainnat}
\bibliography{bibs}

\end{document}